\documentclass{article}
\usepackage{graphicx}
\usepackage{amsmath, amssymb}
\usepackage{amsthm}
\usepackage{tikz}
\usepackage{float}
\usetikzlibrary{trees, positioning,calc}
\usepackage{enumitem}
\usepackage{thm-restate}

\usepackage[style=alphabetic, doi=false,isbn=false,url=false,eprint=false,maxbibnames=99]{biblatex}
\usepackage[colorlinks=true, linkcolor=blue, citecolor=blue, urlcolor=blue]{hyperref}

\newtheorem{claim}{Claim}
\newtheorem{fact}{Fact}
\newtheorem{theorem}{Theorem}
\newtheorem{lemma}{Lemma}

\newtheorem{proposition}{Proposition}

\newcommand{\simv}{\operatorname{sim}}
\newcommand{\mim}{\operatorname{mim}}
\newcommand{\omim}{\operatorname{omim}}
\newcommand{\Omim}{\operatorname{Omim}}

\newcommand{\simw}{\operatorname{sim-width}}
\newcommand{\mimw}{\operatorname{mim-width}}
\newcommand{\lmw}{\operatorname{linear-mim-width}}
\newcommand{\lsw}{\operatorname{linear-sim-width}}
\newcommand{\lomimw}{\operatorname{linear-omim-width}}
\newcommand{\lOmimw}{\operatorname{linear-Omim-width}}
\newcommand{\omimw}{\operatorname{omim-width}}
\newcommand{\Omimw}{\operatorname{Omim-width}}
\newcommand{\lOmw}{\operatorname{linear-Omim-width}}

\newcommand{\uqc}{\operatorname{\textsc{UQC}}}

\bibliography{biblio.bib}

\title{On the hardness of recognizing graphs of small mim-width and its variants}

\author{Max Dupr\'e la Tour\thanks{McGill University: \texttt{maxduprelatour@gmail.com}} \and 
Manuel Lafond\thanks{Universit\'e de Sherbrooke: \texttt{manuel.lafond@usherbrooke.ca}} \and 
Ndiam\'e Ndiaye\thanks{McGill University. \texttt{ndiame.ndiaye@mail.mcgill.ca}} }
\date{}
 
\begin{document}

\maketitle

\begin{abstract}
    The mim-width of a graph is a powerful structural parameter that, when bounded by a constant, allows several hard problems to be polynomial-time solvable - with a recent meta-theorem encompassing a large class of problems [SODA2023].
Since its introduction, several variants such as sim-width and omim-width were developed, along with a linear version of these parameters.
It was recently shown that mim-width and all these variants all paraNP-hard, a consequence of the NP-hardness of distinguishing between graphs of linear mim-width at most 1211 and graphs of sim-width at least 1216 [ICALP2025].  The complexity of recognizing graphs of small width, particularly those close to $1$, remained open, despite their especially attractive algorithmic applications.  

In this work, we show that the width recognition problems remain NP-hard even on small widths.  Specifically, after introducing the novel parameter Omim-width sandwiched between omim-width and mim-width, we show that: (1) deciding whether a graph has sim-width = 1, omim-width = 1, or Omin-width = 1 is NP-hard, and the same is true for their linear variants; (2) the problems of deciding whether mim-width $\leq$ 2 or linear mim-width $\leq$ 2 are both NP-hard.  Interestingly, our reductions are relatively simple and are from the Unrooted Quartet Consistency problem, which is of great interest in computational biology but is not commonly used (if ever) in the theory of algorithms.

\end{abstract}

\section{Introduction}

Many hard problems on graphs become easier when the underlying graph has a simple structure, and over the years several “width’’ parameters have been developed to capture such structure. One of these is the mim-width (maximum induced matching width), introduced by Vatshelle~\cite{vatshelle2012new}. A branch decomposition of a graph $G$ is a ternary tree whose leaves represent the vertices of $G$, and each of its edges naturally defines a cut of the graph. For any such decomposition, each cut induces a bipartite graph, in which we consider the size of a maximum induced matching. The mim-width of $G$ is then the minimum, over all branch decompositions of $G$, of the maximum size of such an induced matching across all cuts of the decomposition~\cite{vatshelle2012new,belmonte2013graph}.

%The \emph{mim-width} (maximum induced matching width) of a graph $G$  bounds the maximum size of an induced matching formed by edges traversing any cut obtainable from a branch decomposition of $G$~\cite{vatshelle2012new,belmonte2013graph} \ml{(a branch decomposition of $G$ is a ternary tree whose leaves are $V(G)$)}. 

It is a powerful parameter that can be constant even on graphs of high clique-width, and several algorithmic problems are polynomial-time solvable on graphs of constant mim-width~\cite{jaffke2020mim,jaffke2020mimii,jaffke2019mimiii}, under the assumption that a branch decomposition is given.  This is notably witnessed by a recent logic-based meta-theorem that solves many problems in polynomial-time under these assumptions~\cite{bergougnoux2023logic}.

Since its inception, several variants of mim-width were developed: the \emph{sim-width} parameter is the same, except that edges of $G$ within both sides of the cut are kept~\cite{kang2017width,munaro2023algorithmic}; the \emph{one-sided mim-width} ($\omimw$) keeps only the edges within either side of the cut and takes the minimum~\cite{bergougnoux2023new}; the \emph{large} one-sided mim-width ($\Omimw$), which we introduce in this work, instead takes the maximum; and the \emph{linear} variant of all these parameters require the branch decomposition to be a caterpillar (see precise definitions below).  

Recognizing graphs of mim-width at most $k$ was quickly established to be W[1]-hard in parameter $k$~\cite{saether2016hardness}, but the question of XP membership was posed several times since the creation of the parameter (in addition to the previous references, this was also raised in~\cite{bergougnoux2022node,bergougnoux2023new,otachi_et_al:LIPIcs.SWAT.2024.38,bergougnoux2023logic}).  The same remained open for all the aforementioned variants of mim-width, except the the linear sim-width = 1 case shown NP-hard in~\cite{Ziedan2018}.  All these questions were finally answered by Bergougnoux, Bonnet, and  Duronin~\cite{bergougnoux_et_al:LIPIcs.ICALP.2025.25}, where the authors showed that it is NP-hard to distinguish between graphs of linear mim-width at most 1211 and graphs of sim-width at least 1216, implying the paraNP-hardness of $\mimw$ but also of all the above variants.

It is still open whether there is a constant $c < 1211$ such that recognizing graphs of mim-width (and variants) at most $c$ is in P.  Determining whether such a $c$ exists is quite relevant, because graphs of small width are those that matter for algorithmic applications.  
In this work, we reduce this knowledge gap significantly by proving the following results:

\begin{restatable}{theorem}{thmsim}\label{thm:sim}
    The problem of deciding whether a graph has sim-width, omim-width, or Omim-width equal to $1$ is NP-complete. Moreover, assuming the ETH, it cannot be solved in time $2^{o(n)}$, where $n$ is the number of vertices of the graph. The same holds for the linear variant of these three parameters.
\end{restatable}

\begin{restatable}{theorem}{thmmim}\label{thm:mim}
    The problem of deciding whether a graph has mim-width or linear mim-width at most 2 is NP-complete. Moreover, assuming the ETH, it cannot be solved in time $2^{o(n)}$, where $n$ is the number of vertices of the graph.
\end{restatable}

Note that the ETH refers to the Exponential Time Hypothesis~\cite{impagliazzo2001complexity}.
We also point out that the new parameter $\Omimw$ is a lower bound on $\mimw$, and an upper bound on $\omimw$, which in turn is an upper bound on $\simw$.

Our reductions are relatively simple and arguably less involved than in~\cite{bergougnoux_et_al:LIPIcs.ICALP.2025.25}, which require a series of reductions with complex gadgets.
Another unique aspect of our reductions is that they are all from the Unrooted Quartet Consistency (\textsc{UQC}) problem, shown NP-hard in~\cite{Steel1992}.  The input to \textsc{UQC} is a set of trees on four leaves, and one must decide whether there is a single tree that contains all of them.

This is not a common problem to reduce from, as \textsc{UQC} is often more seen as an ``end-user'' problem, in the sense that it is mostly popular for its bioinformatics applications.  Its hardness has mainly been used to justify heuristics and approximations, and 
we are not aware of \textsc{UQC} being used to show the hardness of other problems (except for some that obviously contain \textsc{UQC} as a special case~\cite{bryant2001constructing,reaz2014accurate}).
On the other hand, our approach shows that \textsc{UQC} may be well-suited for problems that require finding a branch decomposition, since \textsc{UQC} is a tree reconstruction problem.  It is possible that \textsc{UQC} or similar variants may be useful for other open XP membership problems that involve graph-to-tree representations (for instance module-width~\cite{rao2008clique,belmonte2013graph} which is within a factor two of clique-width, or other similar problems stated in~\cite{hogemo2025mapping}).  

Although our results do not imply hardness for mim-width and variants between 3 and 1210, they indicate that there is little hope of recognizing graphs of small width efficiently.  The notable cases of mim-width 1 and linear mim-width 1 remain open.  These two graph classes have exploitable structural properties, for instance graphs of mim-width 1 are perfect~\cite{vatshelle2012new} and graphs of linear mim-width 1 have no so-called asteroidal edge triples~\cite{hogemo2025mapping}, and their membership in P is still possible. Recognizing \emph{leaf powers}, one of the most natural subclasses of graphs of mim-width 1, was recently shown to be NP-hard~\cite{LeafPowersSODA}.  On the other hand, recognizing \emph{linear} leaf powers, a subclass of graphs of linear mim-width 1, is in P~\cite{LinearLPinP}.

\section{Preliminaries}

\subsection{Notations and Definitions}

All graphs considered are simple and undirected. The vertex set and edge set of a graph $G$ are denoted by $V(G)$ and $E(G)$, respectively. An edge between two vertices $u$ and $v$ is denoted by $uv$ (or equivalently $vu$). For a subset $X \subseteq V(G)$, let $G[X]$ denote the subgraph of $G$ induced by $X$, and let $E_G[X]$ denote the edge set of $G[X]$. For a set of edges $F \subseteq E(G)$, we write $G - F$ for the graph obtained from $G$ by deleting all edges in $F$.

A \emph{matching} of $G$ is a set of edges no two of which share a common endpoint. An \emph{induced matching} of $G$ is a matching $M$ such that every edge of $G$ is incident with at most one edge of $M$. A \emph{cut} of $G$ is a bipartition $(A,B)$ of $V(G)$. We denote by $G[A,B]$ the bipartite subgraph of $G$ with edge set $\{\, uv \in E(G) \mid u \in A,\, v \in B \,\}$.

Unless stated otherwise, all trees in this work are unrooted.  A tree is \emph{ternary} if each vertex has either $1$ or $3$ neighbors\footnote{Note, ternary trees are sometimes called binary, but here we prefer ``ternary'' to emphasize that internal vertices have 3 neighbors, not 2.}.
A \emph{caterpillar} is a tree in which the set of internal vertices induces a path, called its \emph{spine}.
Given a total order $\leq$ on a set $P$ of $n \ge 4$ points, we define the \emph{ternary caterpillar realizing} $\leq$,  with leaf set $P$ as follows. If $p_1 < p_2 < \dots < p_{n-1} < p_n$, then the interval $[2, n-1]$ forms the spine of the caterpillar. Attach each leaf $p_i$ to the integer $i$, and additionally attach $p_1$ to $2$ and $p_n$ to $n-1$. Note that a given ternary caterpillar with leaf set $P$ realizes several total orders, since we may choose which end of the spine is the minimum/maximum and, at each end, which of the two incident leaves plays the “endpoint” role.

\paragraph{Width parameters:}
Given a cut $(A,B)$:
The \emph{mim-value} of $(A,B)$, denoted $\mim_G(A,B)$, is the maximum size of an induced matching in $G[A,B]$.

The \emph{sim-value} of $(A,B)$, denoted $\simv_G(A,B)$, is the maximum size of a matching in $G[A,B]$ that is also an induced matching of $G$ (hence, edges within $A$ and within $B$ are also considered).

For a subset $X \subseteq V(G)$, the \emph{upper-induced matching number} of $X$ is the maximum size of a matching in $G[X, V(G) \setminus X]$ that is an induced matching of the graph $G - E_G[V(G) \setminus X]$.

The \emph{omim-value} of a cut $(A,B)$, denoted $\omim_G(A,B)$, is the minimum of the upper-induced matching numbers of $A$ and $B$.

The \emph{Omim-value} of a cut $(A,B)$, denoted $\Omim_G(A,B)$, is the maximum of the upper-induced matching numbers of $A$ and $B$.

Note that it directly follows from the definitions of the parameters that $\mim_G(A,B) \geq \Omim_G(A,B) \geq \omim_G(A,B) \geq \simv_G(A,B)$.

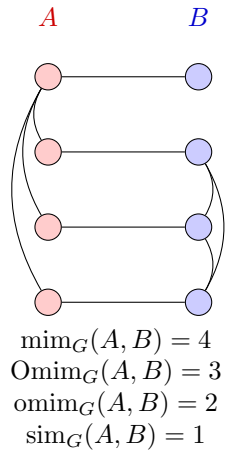
\begin{figure}[H]
    \centering
    \begin{tabular}{c}
        \begin{tikzpicture}[scale=1, every node/.style={circle, draw, minimum size=0.35cm}]
            \node[fill=red!20] (A1) at (-1,3) {};
            \node[fill=red!20] (A2) at (-1,2) {};
            \node[fill=red!20] (A3) at (-1,1) {};
            \node[fill=red!20] (A4) at (-1,0) {};
            \node[fill=blue!20] (B1) at (1,3) {};
            \node[fill=blue!20] (B2) at (1,2) {};
            \node[fill=blue!20] (B3) at (1,1) {};
            \node[fill=blue!20] (B4) at (1,0) {};
            \node[draw=none, fill=none, text=red!80!black, font=\bfseries] at (-1,3.8) {$A$};
            \node[draw=none, fill=none, text=blue!80!black, font=\bfseries] at (1,3.8) {$B$};
            \draw (B1) -- (A1);
            \draw (B2) -- (A2);
            \draw (B3) -- (A3);
            \draw (B4) -- (A4);
            \draw (A1) to [bend right = 30] (A2);
            \draw (A1) to [bend right = 30] (A3);
            \draw (A1) to [bend right = 30] (A4);
            \draw (B2) to [bend left = 30] (B3);
            \draw (B3) to [bend left = 30] (B4);
            \draw (B4) to [bend right = 30] (B2);
        \end{tikzpicture} \\
        $\mathrm{\mim}_G(A,B) = 4$ \\
        $\mathrm{\Omim}_G(A,B) = 3$ \\
        $\mathrm{\omim}_G(A,B) = 2$ \\
        $\mathrm{\simv}_G(A,B) = 1$
    \end{tabular}

\caption{A graph illustrating the differences of the width parameters on a specific cut $(A, B)$}
\end{figure}
A \emph{branch decomposition} of a graph $G$ is a ternary tree $T$ whose leaves correspond exactly to the vertices of $G$.
Given such a tree $T$, each edge $e$ of $T$ defines a cut $(A_e, B_e)$, where $A_e$ and $B_e$ are the vertex sets corresponding to the two connected components of $T - \{e\}$. The \emph{mim-width} (resp. \emph{sim-width}, \emph{omim-width}, \emph{Omim-width}) of $T$ is the maximum, over all edges $e$ of $T$, of $\mim_G(A_e, B_e)$ (resp. $\simv_G(A_e, B_e)$, $\omim_G(A_e, B_e)$, $\Omim_G(A_e, B_e)$). Finally, the \emph{mim-width} (resp. \emph{sim-width}, \emph{omim-width}, \emph{Omim-width}) of $G$ is the minimum, over all branch decompositions $T$ of $G$, of the mim-width (resp. sim-width, omim-width, Omim-width) of $T$. 

Each of these width parameters has a \emph{linear} variant, which require $T$ to be a \emph{caterpillar}.  We denote such a variant by prefixing the parameter name with ``linear'', e.g., linear-mim-width$(G)$ is the minimum mim-width of all branch decompositions of $G$ that are caterpillars.   

\begin{lemma}\label{lem:mimbounds}
    For any graph $G$, we have 
$$
\mimw(G) \geq \Omimw(G) \geq \omimw(G) \geq \simw(G).
$$
This chain of inequalities also holds by replacing each parameter by their linear variant.  Moreover, $\lmw(G) \geq \mimw(G)$, and similarly for each variant of the width parameter.

\end{lemma}

\subsection{Unrooted Quartet Consistency}

Given a set of points $P$, a \emph{quartet} consists of two unordered pairs of elements of $P$ that do not intersect.  For clarity, such a quartet $q = \{ \{p_i, p_j\}, \{p_k, p_{\ell} \} \}$ is instead denoted $q = [p_i p_j | p_k p_{\ell}]$.  Slightly abusing notation, we may write $p \in q$ when $p$ is one of the four elements involved in $q$.  

In a ternary tree $T$ with leaves labeled $\{p_1, \dots, p_n\}$, consider any four leaves $p_i, p_j, p_k, p_{\ell}$.  
If there exists an edge $e$ of $T$ such that $p_i$ and $p_j$ lie on one side of $T - \{e\}$, and $p_k$ and $p_{\ell}$ lie on the other, we denote this relation by  
$$
p_i p_j \mid_T p_k p_{\ell}.
$$
In this case, we say that $T$ \emph{satisfies} the quartet $q = [p_i p_j \mid p_k p_{\ell}]$.
If $Q$ is a set of quartets, then $T$ satisfies $Q$ if it satisfies each $q \in Q$.

Because $T$ is ternary, given four leaves $p_i, p_j, p_k, p_\ell$, exactly one of the three configurations holds:
$$
p_i p_j \mid_T p_k p_{\ell}, \quad
p_i p_k \mid_T p_j p_{\ell}, \quad \text{or} \quad
p_i p_{\ell} \mid_T p_j p_k.
$$

%When the choice of the tree is clear from context, we will omit the subscript $T$ and simply write $p_i p_j \mid p_k p_{\ell}$ instead of $p_i p_j \mid_T p_k p_{\ell}$.

The NP-hard problem used in our reductions is the \emph{Unrooted Quartet Consistency} (\textsc{UQC}) problem, introduced by Steel~\cite{Steel1992}.

\medskip 

\noindent
\textbf{Input:} A set $Q$ of quartets and a set of points $P = \{p_1, p_2, \ldots, p_n\}$.\\
\textbf{Question:} Does there exist a ternary tree $T$ with leaves labeled by the points in $P$ that satisfies every quartet in $Q$?

\medskip

The hardness of \textsc{UQC} is proved by a reduction from \textsc{Betweenness}, a classical problem that was shown to be NP-complete by Opatrný~\cite{Opatrny1979}.  In that problem, the input is a set $B$ of ordered triples whose elements are from a set of points $U$.  One must decide whether there is a total order $<$ on $U$ such that, for each $(a, b, c) \in B$, one of $a < b < c$ or $c < b < a$ holds.  Steel's reduction proves a slightly stronger result than NP-hardness, which we will use.

\begin{theorem}[\cite{Steel1992}]\label{thm:steel}
There is a polynomial-time reduction from \textsc{Betweenness} to \textsc{UQC}. Moreover, if an instance of \textsc{Betweenness} is a YES-instance, then there exists a \emph{caterpillar} that satisfies the corresponding \textsc{UQC} instance.
\end{theorem}

By analyzing the chain of reductions used to prove Theorem~\ref{thm:steel}, we can also obtain a slightly stronger result under ETH.

\begin{theorem}\label{thm:uqceth}
    Assuming the ETH, the \textsc{UQC} problem cannot be solved in time $2^{o(n+m)}$, where $n$ is the number of points of an instance and $m$ is the number of quartets.
\end{theorem}

\begin{proof}
    The result can be deduced by simply revisiting the chain of reductions resulting in the NP-hardness of \textsc{UQC}.  
    The first reduction is from \textsc{3-Set Splitting} to \textsc{Betweenness}.  In \textsc{3-Set Splitting}, we receive sets $S$ of size 3 over a universe $U$ and must color $U$ with two colors so that no set of $S$ is monochromatic.  This cannot be solved in time $2^{o(|U| + |S|)}$ under the ETH (a proof appears in~\cite[Proposition 5.1]{antony2024switching}).

    Opatrný's reduction from \textsc{3-Set Splitting} to \textsc{Betweenness}~\cite[Lemma 2]{Opatrny1979} produces from $(S, U)$ an instance $(B, U_B)$ with $|B| = 3m$ and $|U_B| = 2n + 1$ (each set in $S$ becomes three betweenness triples, each element of $U$ has two corresponding elements in $U_B$, plus an extra point).  It follows that \textsc{Betweenness} cannot be solved in time $2^{o(|B| + |U_B|)}$.

    Finally, Steel's reduction~\cite{Steel1992} from  \textsc{Betweenness} to \textsc{UQC} transforms $(B, U_B)$ to an instance $(Q, P)$ with $|Q| = 6|B|$ (each betweenness triple becomes 6 quartets) and $|P| = |U_B| + 2 + 4|B|$ (the point set consists of the elements of $U_B$ and two extra points, and each betweenness constraint adds four points).  Since $|P| + |Q|$ is linear in $|B| + |U_B|$, \textsc{UQC} cannot be solved in time $2^{o(|P| + |Q|)} = 2^{o(n + m)}$.
\end{proof}

\section{Sim-width, omim-width and Omim-width}

We start with a single reduction  showing that it is NP-hard to decide whether a given graph has sim-width, omim-width, or Omim-width equal to $1$.  The same holds for their linear variants.
Given a \textsc{UQC} instance $(Q, P)$, we define the following graph $G$:

\begin{itemize}
    \item \textbf{Vertices:} For each element $p \in P$, we introduce a vertex $p$ in $G$. 
    For each quartet $q = [p_i p_j \mid p_k p_{\ell}] \in Q$, we introduce four additional vertices in $G$:
    \[
        \{u_i^q, u_j^q, u_k^q, u_{\ell}^q\}.
    \]
    Let $U$ denote the set of all such $u$-vertices, so that $|U| = 4|Q|$.

    \item \textbf{Edges:} For each quartet $q = [p_i p_j \mid p_k p_{\ell}] \in Q$, add the edges
    \[
        p_iu_i^q, u_i^q u_j^q, u_j^q p_j, \text{ and }
        p_k u_k^q, u_k^q u_{\ell}^q, u_{\ell}^q p_{\ell}.
    \]
    In addition, connect every pair of vertices $u, u' \in U$ whenever they correspond to different quartets.   That is, for $u_i^q, u_j^r \in U$, we add the edge $u_i^q u_j^r$ whenever $q \neq r$.
    
    Observe that $P$ forms an independent set, and that each $u_i^q \in U$ has only two non-neighbors in $U$.  
\end{itemize}

\begin{figure}[h]
\centering
\begin{tikzpicture}[
    scale=1,
    every node/.style={circle, draw, minimum size=0.8cm},
    u/.style={circle, draw, fill=gray!20, minimum size=0.8cm},
    box/.style={rectangle, draw, rounded corners, minimum width=2.6cm, minimum height=1.2cm, align=center}
]
% Top row: the P-vertices
\node (pi)  at (0,2) {$p_i$};
\node (pj)  at (2,2) {$p_j$};
\node (pk)  at (4,2) {$p_k$};
\node (pl)  at (6,2) {$p_{\ell}$};

% Quartet-specific u-vertices
\node[u] (ui) at (0,0) {$u_i^q$};
\node[u] (uj) at (2,0) {$u_j^q$};
\node[u] (uk) at (4,0) {$u_k^q$};
\node[u] (ul) at (6,0) {$u_{\ell}^q$};

% Required edges for this quartet
\draw (pi) -- (ui);
\draw (ui) -- (uj);
\draw (uj) -- (pj);

\draw (pk) -- (uk);
\draw (uk) -- (ul);
\draw (ul) -- (pl);

\end{tikzpicture}
\caption{Subgraph of $G$ corresponding to a quartet $q = [p_i p_j \mid p_k p_{\ell}]$.}
\end{figure}
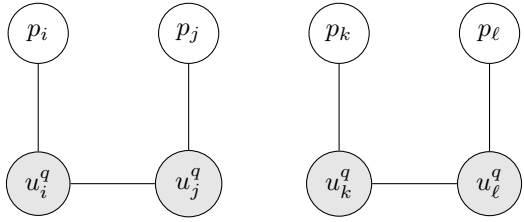

Note that $|V(G)| = |P| + 4|Q|$. The fact that $|V(G)|$ is linear in $|P|$ and $|Q|$ will be used to obtain the ETH-based hardness part of our result. Note that we are not able to obtain the stronger lower bound under ETH of $2^{o(|E(G)|)}$, because the number of edges in this reduction is quadratic.

The combination of Lemma~\ref{lem:mimbounds}, Theorem~\ref{thm:uqceth}, and the following proposition will imply our first hardness result, Theorem~\ref{thm:sim}.

\begin{proposition}\label{prop:swreduction}
    If there is a caterpillar satisfying the \textsc{UQC} instance $(Q,P)$, then $\lOmw(G) = 1$. Otherwise, if $(Q, P) \notin \uqc$, then any branch decomposition $T$ of $G$ satisfies $\simw_G(T) = 2$, and in particular $\simw(G) = 2$.
\end{proposition}

We split the proof in two lemmas.

\begin{lemma}\label{lem:sw2}
    If $(Q, P) \notin \uqc$, then any branch decomposition $T$ of $G$ satisfies $\simw_G(T) = 2$, and in particular $\simw(G) = 2$.
\end{lemma}
\begin{proof}
    We start by proving that $\simw(G) \leq 2$. 
    
    This follows directly from the fact that the maximum induced matching in the full graph $G$ has size at most 2. Suppose, for contradiction, that there exists an induced matching of size at least $3$. Since the set $P$ is an independent set in $G$, at least one endpoint of every edge in the matching must belong to $U$.
    By construction, each vertex in $U$ is adjacent to all other vertices of $U$ except two, which themselves are adjacent. Hence, among any three such vertices of $U$, there must exist at least one edge between them, contradicting the assumption that they form an induced matching.

    We now prove that if $(Q, P) \notin \uqc$, then $\simw(G) \ge 2$.

    Let $(Q, P) \notin \uqc$, and let $T$ be a branch decomposition of $G$. There must exist a quartet $q = [p_i p_j \mid p_k p_{\ell}] \in Q$ not satisfied by $T$, otherwise $(Q, P)$ would be in \textsc{UQC}. Assume without loss of generality that $p_i p_k \mid_T p_j p_{\ell}$, and let $e$ be an edge of $T$ inducing the cut $(A_e, B_e)$ such that $p_i, p_k \in A_e$ and $p_j, p_{\ell} \in B_e$.
    At least one of the edges along the path $p_i- u_i^q -u_j^q- p_j$ must cross this cut, and similarly, one of the edges along the path $p_k- u_k^q- u_{\ell}^q- p_{\ell}$ must also cross the cut.  Because there is no edge with one endpoint in the first path and the other in the second path, 
    these two edges form an induced matching in $G$, and therefore $\simw_G(T) \ge 2$ for every branch decomposition $T$. Consequently, $\simw(G) \ge 2$.
\end{proof}

\begin{lemma}\label{lem:lomw1}
    If there is a caterpillar satisfying the \textsc{UQC} instance $(Q,P)$, then $\lOmw(G) = 1$.
\end{lemma}

\begin{proof}

    Let $T$ be such a caterpillar, and let $\leq$ be a total order on $P$ that $T$ realizes. For each $p_i \in P$, define $$C_i = \{p_i\} \cup \{u_i^q \mid q \in Q,\; p_i \in q\},$$ that is, $C_i$ consists of all vertices in $P \cup U$ whose index is $i$. We extend $\leq$ to a total order $\leq'$ on $P \cup U$ as follows. 

    For all $p_i, p_j \in P$ with $p_i < p_j$, and for all $x \in C_i$, $y \in C_j$, we set $x <' y$. Within each $C_i$, the relative position of $p_i$ and any $u_i^q$ depends on the quartet $q$:

    If $q = [p_i p_j \mid p_k p_\ell]$, then, since $T$ satisfies $q$, we have either $p_i, p_j \leq p_k, p_\ell$ or $p_k, p_\ell \leq p_i, p_j$. In the former case ($p_i, p_j \leq p_k, p_\ell$) we set $p_i <' u_i^q$, in the latter case ($p_k, p_\ell \leq p_i, p_j$) we set $u_i^q <' p_i$. Note that this implies that we have either $$p_i <' u_i^q <' u_k^q, u_\ell^q, \text{ or } u_k^q, u_\ell^q <' u_i^q <' p_i.$$

    We choose an arbitrary order within $C_i$ among the $u$-vertices to the left of $p_i$ and among those to the right of $p_i$. We then consider the ternary caterpillar $T'$ realizing $\leq'$. We now prove that $\Omimw_G(T') = 1$, that is, $\Omim_G(A_e, B_e) = 1$ for every edge $e$ of $T'$.

    Because $T'$ is a caterpillar, it suffices to consider the cuts induced by the edges of its spine. Let $e$ be an edge of the spine, and let $(A_e, B_e)$ be the cut induced by $e$. Since $T'$ realizes the order $\leq'$, one side of the cut corresponds to an initial segment of $\leq'$. Thus, we have either for all $a \in A_e$ and $b \in B_e$, $a <' b$, or for all $a \in A_e$ and $b \in B_e$, $b <' a$.

    To prove that $\Omim_G(A_e,B_e) = 1$, we need to prove that the upper induced matching numbers of $A_e$ and $B_e$ are both equal to one. This means that there is no matching of size $2$ in $G[A_e,B_e]$ that is an induced matching in the graph $G - E_G[A_e]$ or in the graph $G- E_G[B_e]$. Let $a_1b_1, a_2b_2 \in E(G)$ be two edges with disjoint endpoints such that $a_1, a_2 \in A_e$ and $b_1, b_2 \in B_e$.  To prove that it is not an induced matching in $G - E_G[A_e]$ or $G - E_G[B_e]$, it suffices to show that one of the two following conditions holds:

    \begin{enumerate}
        \item One of the diagonals $a_1b_2$ or $a_2b_1$ is an edge, or\label{condition:1}
        \item Both $a_1a_2$ and $b_1b_2$ are edges. \label{condition:2}
    \end{enumerate}
    
    Each vertex $p_i \in P$ is connected only to vertices of the form $u_i^q$. Consider two edges $p_i u_i^q$ and $p_j u_j^{q'}$ with $i \neq j$. By the construction of the order $\leq'$, we have either $p_i, u_i^q <' p_j, u_j^{q'}$ or $p_j, u_j^{q'} <' p_i, u_i^q$. Therefore, it is impossible for the cut $(A_e, B_e)$ to contain both edges $p_i u_i^q$ and $p_j u_j^{q'}$, as the cut can separate at most one $p_i$ vertex from one of its neighbors in $G$. Consequently, at most one vertex among $\{a_1, a_2, b_1, b_2\}$ can belong to $P$.

    Suppose that exactly one vertex among $\{a_1, a_2, b_1, b_2\}$ belongs to $P$. Without loss of generality, assume that $a_1 = p_i$. Then $b_1$ must be of the form $u_i^q$. Suppose by contradiction that condition~\ref{condition:1} is false and that $a_2b_1$ is not an edge. Noting that $b_1 \in U$, the only vertices in $U$ that are not adjacent to $u_i^q$ are those corresponding to the other side of the same quartet. That is, if $q = [p_i p_j \mid p_k p_\ell]$, then these vertices are $u_k^q$ and $u_\ell^q$. By the definition of the order $\leq'$, we have either $p_i <' u_i^q <' u_k^q, u_\ell^q$, or $u_k^q, u_\ell^q <' u_i^q <' p_i$. Since the edge $e$ lies on the path between $a_1 = p_i$ and $b_1 = u_i^q$ in $T'$, the vertices $u_k^q$ and $u_\ell^q$ cannot belong to $A_e$. This yields a contradiction.

    Now suppose that all vertices $\{a_1, a_2, b_1, b_2\}$ belong to $U$. Let $q$ be the quartet associated with $a_1$ and $q'$ the quartet associated with $b_1$. As before, assume for contradiction that condition~\ref{condition:1} is false, and that neither $a_1b_2$ nor $a_2b_1$ is an edge. The only vertices in $U$ not adjacent to $a_1$ correspond to the other side of the same quartet, so $q$ must also be the quartet of $b_2$. Similarly, $q'$ must be the quartet of $a_2$.

    If $q = q'$, then we have $\{a_1, a_2, b_1, b_2\} = \{u_i^q, u_j^q, u_k^q, u_\ell^q\}$ for some $q = [p_i p_j \mid p_k p_\ell] \in Q$. The only edges among these vertices are $u_i^q u_j^q$ and $u_k^q u_\ell^q$, which are both separated by $e$. Consequently, the quartet $[u_i^q u_j^q \mid u_k^q u_\ell^q]$ is not satisfied by $T'$, which, by the definition of $\leq'$, implies that $q$ is not satisfied by $T$, and we have a contradiction.\\
    On the other hand, if $q \neq q'$, then both $a_1a_2$ and $b_1b_2$ are edges, and hence condition~\ref{condition:2} holds.
\end{proof}

Proposition~\ref{prop:swreduction} follows from Lemma~\ref{lem:sw2} and Lemma~\ref{lem:lomw1}.
We are now ready to prove Theorem~\ref{thm:sim}.

\thmsim*

\begin{proof}
    Let $\sigma$ be one of the parameters $\simw, \omimw, \Omimw$, $\lsw, \lomimw$, or $\lOmimw$.  We want to show that deciding whether $\sigma(G) = 1$ is NP-complete. Membership in NP is easy since a branch decomposition is an easily verifiable certificate.  
    
    For NP-hardness, for technical accuracy we need to use Theorem~\ref{thm:steel} and reduce \textsc{Betweenness} to deciding $\sigma(G) = 1$. 
    Take an instance $I$ of \textsc{Betweenness} and let $(Q, P)$ be the \textsc{UQC} instance obtained from Theorem~\ref{thm:steel}, and then obtain $G$ as constructed above from that \textsc{UQC} instance.  If $I$ is a YES-instance of \textsc{Betweenness}, then $(Q, P) \in $ \textsc{UQC} and by Theorem~\ref{thm:steel} there is a ternary caterpillar that satisfies $Q$.  Then, Proposition~\ref{prop:swreduction} implies that $\lOmimw(G) = 1$.  This implies that $\sigma(G) = 1$, since $\lOmimw(G)$ is an upper bound on all the possible $\sigma$ parameters listed above (by Lemma~\ref{lem:mimbounds}).

    Conversely, if $I$ is a NO-instance of \textsc{Betweenness}, then no tree satisfies the \textsc{UQC} instance $(Q, P)$.  By Proposition~\ref{prop:swreduction}, $\simw(G) = 2$, which implies $\sigma(G) \geq 2$ since $\simw(G)$ is a lower bound on all the possible $\sigma$ parameters listed above (by Lemma~\ref{lem:mimbounds}).
    Therefore, $I$ is a YES-instance of \textsc{Betweenness} if and only if $\sigma(G) = 1$.

    Moreover $|V(G)| =  |P| + 4|Q|$, therefore an algorithm recognising graphs with $\sigma$ parameter equal to one in time $2^{o(|V(G)|)}$ could be used to solve \textsc{UQC} in time $2^{o(|P|+|Q|)}$, contradicting Theorem~\ref{thm:uqceth} assuming ETH.
\end{proof}

\paragraph{Remarks on the graph class of hard instances.}  Let us observe that, denoting by $C_k$ the chordless cycle on $k$ vertices, our construction produces a graph $G$ in which the $U$ vertices induce the complement of a set of disjoint $C_4$'s, and that all $P$ vertices have only neighbors in $U$ and are simplicial i.e., their neighbors form a clique (since two distinct neighbors of a $p_i$ vertex have the form $u_i^q, u_i^{q'}$ with $q \neq q'$).  This implies that $G$ is $C_k$-free for every $k \geq 5$: the $p_i$'s cannot be in a $C_k$ because they are simplicial, and and can show that the $U$ vertices cannot have a $C_k$ since they are the complement of disjoint $C_4$'s.  We note that for chordal graphs, i.e., graphs that are $C_k$-free for every $k \geq 4$, whether $\lsw(G) = 1$ can be decided in polynomial time~\cite{Ziedan2018}. The set of forbidden induced cycles from our construction is therefore, in some sense, tight for the casse $\lsw$ equal to 1.  

Additionally, the complement $\overline{G}$ is also $C_k$-free for $k \geq 5$.  To see this, suppose that $\overline{G}$ has an induced $C_k$ and notice that it must contain some $p_i$ vertex (as the $U$ vertices induce disjoint $C_4$'s).  Since $P$ is a clique in $\overline{G}$, $p_i$ must have some neighbor $u_j^q$ on $C_k$, with $i \neq j$.  The other neighbor of $u_j^q$ on $C_k$ can only be $u_i^q$ since a $p_j$ vertex or a $u_\ell^q$ vertex with $i \neq \ell$ would form a triangle with $p_i$.  But then, the other neighbor of $u_i^q$ is either some $u_\ell^q$ or some $p_j$ vertex, which are both neighbors of $p_i$, and so no $C_k$ is possible for $k \geq 5$.  
Thus the hard instances do not contain induced cycles of length 5 or more nor their complements.  In particular, they are perfect graphs and the NP-hardness also holds on this class of graphs. 

\paragraph{Using a similar reduction for mim-width 1 (or linear mim-width 1)?}

It is interesting to ponder the possibility of using a similar reduction for $\mimw=1$.  It appears difficult to do so and here we provide the intuitive reasons for this.
%To obtain a reduction for $\mimw = 1$, we attempted to adapt the same strategy. 
The high-level idea in our reduction is to construct a graph $G$ that includes vertices representing the elements of the UQC instance and, for each quartet $q$, a gadget $H_q$.  A reduction for $\mimw = 1$ would have the following intended properties.

\begin{enumerate}
     \item In any branch decomposition of $G$, if the four vertices of a quartet $q$ do not satisfy $q$, then some edge of the decomposition induces a cut within $H_q$ containing an induced matching of size~2. Thus, $H_q$ is meant to ``force'' a particular shape for the subtree containing these four vertices in any branch decomposition of mim-width~1.

    \item If the UQC instance is satisfiable, then there exists a branch decomposition of $G$ whose mim-width is exactly~1.
\end{enumerate}

The difficulty arises when we try to satisfy these conditions simultaneously. Consider two different quartets $q$ and $q'$ with corresponding gadgets $H_q$ and $H_{q'}$. In the branch decomposition guaranteed by property~2, the subtrees corresponding to these two gadgets may overlap (it seems hard to avoid this). As a consequence, an edge of $H_q$ used in the induced matching when $q$ is unsatisfied (property~1) may lie in the same cut as an edge of $H_{q'}$ used when $q'$ is unsatisfied.

To ensure mim-width~1, we must prevent these two edges to form an induced matching. This forces us to add \emph{blocking edges} between certain endpoints of edges coming from $H_q$ and $H_{q'}$ in order to destroy the unwanted induced matchings. In our reduction, the blocking edges correspond to the edges between the $u$-vertices corresponding to different quartets.

However, these blocking edges create a new problem: they may themselves form an induced matching of size~2 across some cut, contradicting property~2 and preventing the existence of a mim-width~1 decomposition even when the UQC instance is satisfiable. In our reduction, this correspond to the last case in the proof of Lemma~\ref{lem:lomw1}: the matching of the type $u_x^q u_y^{q'}$, $u_w^{q'} u_z^q$ for different $q,q'$, which  is only destroyed with edges not in the cut.

In short, enforcing quartet constraints independently via the gadgets $H_q$ inevitably causes interactions (either among the gadgets or among the blocking edges) that create new induced matchings. We found no way around this, which is why our approach does not seem extendable to mim-width~1 (or linear mim-width~1).

\section{Mim-width and linear mim-width}

In this section, we describe two closely related reductions. The first is used to prove the hardness of recognizing graphs with $\mimw$~2, and the second establishes the hardness of recognizing graphs with $\lmw$~2.

We modify the previous reduction by adding four new vertices per quartet. Additionally, one extra vertex is introduced in the non-linear version.

Given a \textsc{UQC} instance $(Q, P)$, we define the following graph $G$:

\begin{itemize}
    \item \textbf{Vertices:} For each element $p \in P$, we introduce a vertex $p$. 
    For each quartet $q = [p_i p_j \mid p_k p_{\ell}] \in Q$, we introduce eight additional vertices:
    \[
        \{u_i^q, u_j^q, u_k^q, u_{\ell}^q, \gamma_i^q, \gamma_j^q, \gamma_k^q, \gamma_{\ell}^q\}
    \]
    Let $U$ (resp.\ $\Gamma$) denote the set of all such $u$-vertices (resp.\ $\gamma$-vertices).

    \item \textbf{Edges:} For each quartet $q = [p_i p_j \mid p_k p_{\ell}] \in Q$, add the edges
    \[
        \{p_i, u_i^q\}, \{u_i^q, u_j^q\}, \{u_j^q, p_j\}, 
        \{p_k, u_k^q\}, \{u_k^q, u_{\ell}^q\}, \{u_{\ell}^q, p_{\ell}\},
    \]
    together with
    \[
        \{p_i, \gamma_i^q\}, \{p_j, \gamma_j^q\}, 
        \{p_k, \gamma_k^q\}, \{p_{\ell}, \gamma_{\ell}^q\}.
    \]
    Moreover, add all possible edges so that $\Gamma$ forms a clique, add an edge between every pair of vertices $u, u' \in U$ whenever they correspond to different quartets (i.e., whenever their superscript differs), and between every pair of vertices $u \in U$ and $\gamma \in \Gamma$, whenever they correspond to different quartets.
\end{itemize}

The graph $H$ is obtained from $G$ by adding an extra vertex~$\omega$, sharing an edge with every vertex in~$U$.

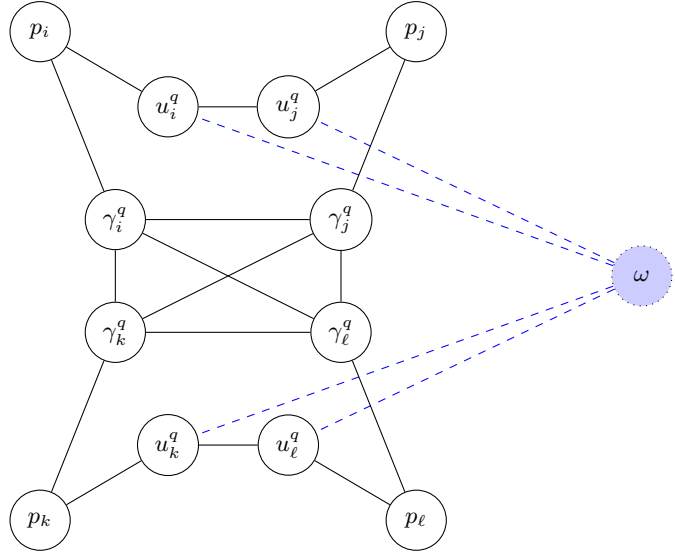
\begin{figure}[h]
\centering
\begin{tikzpicture}[
    scale=1,
    every node/.style={circle, draw, minimum size=0.8cm, font=\small},
    v/.style={},
    u/.style={},
    g/.style={}
    ]

%-----------------------------------
% Top layer: v_i and v_j
%-----------------------------------
\node[v] (vi) at (-1,4) {$p_i$};
\node[v] (vj) at (4,4) {$p_j$};

%-----------------------------------
% Second layer: u_i^q and u_j^q
%-----------------------------------
\node[u] (ui) at (0.7,3) {$u_i^q$};
\node[u] (uj) at (2.3,3) {$u_j^q$};

%-----------------------------------
% Middle layer: gamma nodes (square)
%-----------------------------------
\node[g] (gi) at (0,1.5) {$\gamma_i^q$};
\node[g] (gj) at (3,1.5) {$\gamma_j^q$};
\node[g] (gk) at (0,0) {$\gamma_k^q$};
\node[g] (gl) at (3,0) {$\gamma_{\ell}^q$};

%-----------------------------------
% Lower layer: u_k^q and u_l^q
%-----------------------------------
\node[u] (uk) at (0.7,-1.5) {$u_k^q$};
\node[u] (ul) at (2.3,-1.5) {$u_{\ell}^q$};

%-----------------------------------
% Bottom layer: v_k and v_l
%-----------------------------------
\node[v] (vk) at (-1,-2.5) {$p_k$};
\node[v] (vl) at (4,-2.5) {$p_{\ell}$};

%-----------------------------------
% Omega Vertex
%-----------------------------------
\node[fill=blue!20,dotted] (w) at (7,0.75) {$\omega$};

%-----------------------------------
% Edges: omega
%-----------------------------------
\draw[blue,dashed] (w) -- (ui);
\draw[blue,dashed] (w) -- (uj);
\draw[blue,dashed] (w) -- (uk);
\draw[blue,dashed] (w) -- (ul);

%-----------------------------------
% Edges: upper path
%-----------------------------------
\draw (vi) -- (ui) -- (uj) -- (vj);

%-----------------------------------
% Edges: lower path
%-----------------------------------
\draw (vk) -- (uk) -- (ul) -- (vl);

%-----------------------------------
% Attach v’s to their gammas
%-----------------------------------
\draw (vi) -- (gi);
\draw (vj) -- (gj);
\draw (vk) -- (gk);
\draw (vl) -- (gl);

%-----------------------------------
% Gamma clique (square)
%-----------------------------------
\draw (gi) -- (gj);
\draw (gj) -- (gl);
\draw (gl) -- (gk);
\draw (gk) -- (gi);
\draw (gi) -- (gl);
\draw (gj) -- (gk);

\end{tikzpicture}
\caption{Gadget corresponding to a single quartet $q = [p_i p_j \mid p_k p_{\ell}]$, with the extra vertex $\omega$ of $H$ represented.}
\end{figure}

Note, the constructed graphs $G$ and $H$ have $\mimw$ at least 2, since Vatshelle argued that graphs of $\mimw$ $1$~\cite[Corollary 3.7.4]{vatshelle2012new} have no induced cycle of length 5 or more, whereas these graphs have such cycles of the form $p_i - u_i^q - u_j^q - p_j - \gamma_j^q - \gamma_i^q - p_i$.

\begin{proposition}\label{prop:mimw-lmimw}
    If there is a caterpillar satisfying the \textsc{UQC} instance $(Q,P)$, then $\lmw(G) = 2$, and $\mimw(H) = 2$. Otherwise, if $(Q,P) \notin \uqc$ then $\lmw(G) \geq 3$, and $\mimw(H) \geq 3$.
\end{proposition}

The proof have the same structure as in the previous section, but require a more involved case analysis. We split it in several lemmas.

\begin{lemma}\label{lem:G3}
If $(Q,P) \notin \uqc$ then $\lmw(G) \geq 3$. 
\end{lemma}

\begin{lemma}\label{lem:H3}
If $(Q,P) \notin \uqc$ then $\mimw(H) \geq 3$. 
\end{lemma}

The first step in the proofs of these two lemmas is similar. We begin by showing a result on any branch decomposition of $G$. Any branch decomposition of $H$ can be transformed into one of $G$ by deleting the $\omega$-leaf. The result will therefore also apply to branch decompositions of $H$.

Let $(Q, P) \notin \uqc$, and let $T$ be  
\emph{any} branch decomposition of $G$. Assume by contradiction that $\mim_G(T) \leq 2$. There must exist a quartet $q = [p_i p_j \mid p_k p_{\ell}] \in Q$ that is not satisfied by $T$; otherwise, $(Q, P)$ would belong to $\uqc$. 
    
Without loss of generality, assume that $p_i p_k \mid_T p_j p_{\ell}$ is true. Let $\pi$ be the set of all edges $e$ of $T$ inducing a cut $(A_e, B_e)$ such that $p_i, p_k \in A_e$ and $p_j, p_{\ell} \in B_e$. Note that $\pi$ forms a path. Let $A$ (resp. $B$) denote $\bigcap_{e \in \pi} A_e$ (resp. $\bigcap_{e \in \pi} B_e$); that is, $A$ (resp. $B$) consists of all vertices on the side of $p_i, p_k$ (resp. $p_j, p_{\ell}$) of $\pi$.

\begin{claim}\label{claim:gammas} Either $\gamma_i^q, \gamma_j^q, \gamma_k^q, \gamma_{\ell}^q$ are all contained in $A$, or they are all contained in $B$.
\end{claim}

\begin{proof}    
    Suppose, for the sake of contradiction, that there exists an edge $e$ of $\pi$ such that one of the $\gamma$-vertices lies in $A_e$ and another lies in $B_e$. We will show that we can always find an induced matching of size $3$ in $G[A_e,B_e]$. We proceed by a case analysis.

    First, consider the case where both $\gamma_i^q$ and $\gamma_k^q$ lie in $A_e$. Choose $x \in \{j, \ell\}$ such that $\gamma_x^q$ is a $\gamma$-vertex in $B_e$. Along the path $p_i - u_i^q - u_j^q - p_j$ of $G$, at least one edge crosses the cut $(A_e,B_e)$; likewise, at least one edge of the path $p_k - u_k^q - u_\ell^q - p_\ell$ also crosses the cut. These two crossing edges, together with the edge $\gamma_i^q \gamma_x^q$, form an induced matching in $G[A_e,B_e]$. Indeed, the only other edges that could appear among these six vertices in $G$ are $p_i \gamma_i^q$, which lies entirely within $A_e$, and $p_x \gamma_x^q$, which lies entirely within $B_e$.

    Therefore, we may assume that at least one of $\gamma_i^q$ or $\gamma_k^q$ is in $B_e$.  Without loss of generality, we may assume that $\gamma_i^q \in B_e$. By applying the symmetric argument on the pair $\gamma_j^q, \gamma_\ell^q$, we may assume that one of the vertices $\gamma_j^q$ or $\gamma_\ell^q$ lies in $A_e$.

    Suppose now that $u_j^q \in A_e$. Then one of the edges on the path $p_k - u_k^q - u_\ell^q - p_\ell$, together with the edges $p_i \gamma_i^q$ and $u_j^q p_j$, forms an induced matching in $G[A_e,B_e]$. Hence we may assume $u_j^q \in B_e$.

    Next, suppose that $u_i^q \in A_e$. Then one of the edges on the path $p_k - u_k^q - u_\ell^q - p_\ell$, together with the edges $p_i \gamma_i^q$ and $u_i^q u_j^q$, forms an induced matching in $G[A_e,B_e]$. Hence we may assume $u_i^q \in B_e$.

    %By symmetry, if we had $\gamma_j^q \in A_e$, the preceding arguments would force $u_i^q, u_j^q \in A_e$, a contradiction. 
    If we had $\gamma_j^q \in A_e$, then one of the edges on the path $p_k - u_k^q - u_\ell^q - p_\ell$, together with the edges $p_i u_i^q$ and $\gamma_j^q p_j$, forms an induced matching in $G[A_e, B_e]$.
    Therefore $\gamma_j^q \in B_e$.

    It follows that $\gamma_\ell^q \in A_e$, and by symmetry again, we have $\gamma_k^q, u_k^q, u_\ell^q \in A_e$ as well. In this situation, the three edges $u_\ell^q p_\ell$, $p_i u_i^q$, and $\gamma_k^q \gamma_j^q$ form an induced matching in $G[A_e,B_e]$. 
    
    This establishes that either $\gamma_i^q, \gamma_j^q, \gamma_k^q, \gamma_{\ell}^q$ are all contained in $A$, or they are all contained in $B$. 
\end{proof}

We will assume without loss of generality that all the gammas are in $B$. We are now ready to finish the proof of Lemma~\ref{lem:G3}.

\begin{proof}[Proof of Lemma~\ref{lem:G3}]

We now assume for contradiction that $T$ is a caterpillar satisfying $\lmw(T) \leq 2$, implying in particular that $\mimw(T) \leq 2$. 
Note that Claim~\ref{claim:gammas} is applicable to $T$ since it holds for any branch decomposition. Now, one of $p_j$ or $p_\ell$ is adjacent to an endpoint of $\pi$ in $T$ (recall that $\pi$ is the set of edges separating $A$ and $B$). Without loss of generality, suppose it is $p_j$, and let $z$ be the neighbor of $p_j$ on the spine of the caterpillar. Let $e$ be the edge incident to $z$ that is not in $\pi$, other than $z p_j$ (so, $e$ is the next edge on the spine beyond $\pi$ in the direction past $p_j$, going towards $p_\ell$). This induces a cut $(A_e,B_e)$ with $p_\ell$ and all the $\gamma$-vertices in $B_e$, and $p_i,p_j,p_k \in A_e$. Consequently, the edges $p_i \gamma_i^q $, $p_j \gamma_j^q $, and $p_k \gamma_k^q $ form an induced matching in $G[A_e,B_e]$. This completes the proof.
\end{proof}
    
The proof of Lemma~\ref{lem:H3} is slightly more complicated, because $T$ might not be a caterpillar. Therefore, we need a case analysis depending on which vertices of $B$ lie in the branch containing $p_j$ and which lie in the branch containing $p_\ell$. The additional vertex $\omega$ is required to force an induced matching of size $3$ in all configurations.

\begin{proof}[Proof of Lemma~\ref{lem:H3}]

We assume for contradiction that $T$ is a branch decomposition of $H$ satisfying $\mimw(T) \leq 2$.
Let $v$ be the endpoint of $\pi$ that lies on the path between $p_j$ and $p_\ell$ in $T$ (thus, one incident ``side'' of $v$ contains $p_j$, one contains $p_\ell$, and the remaining side contains both $p_i$ and $p_k$, see figure). Let $e_{mid}$ be the edge from $v$ toward $p_i$ and $p_k$, $e_{top}$ the edge toward $p_j$, and $e_{bot}$ the edge toward $p_\ell$. We denote by $(A_{mid}, B_{mid})$, $(A_{bot}, B_{bot})$, and $(A_{top}, B_{top})$ the cuts induced by $e_{mid}$, $e_{bot}$, $e_{top}$, respectively. We have $B_{mid} = B$, and we partition $B$ as $B = B_{bot} \cup B_{top},$
where $B_{bot}$ (resp.\ $B_{top}$) consists of the vertices of $B$ in the branch toward $p_\ell$ (resp.\ toward $p_j$).

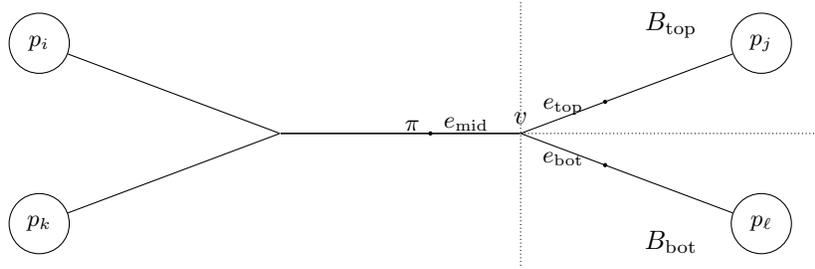
\begin{tikzpicture}[
    scale=0.8,
    leaf/.style={circle, draw, minimum size=0.8cm, font=\small},
    internal/.style={inner sep=0, minimum size=0, draw=none},
    edgelabel/.style={font=\small, inner sep=1pt},
]

% main internal vertices
\node[internal] (u) at (-2,0) {};
\node[internal, label=above:$v$] (v) at (2,0) {};

% leaves (only these are circles) – move them further for longer edges
\node[leaf] (piL) at (-6,1.5) {$p_i$};
\node[leaf] (pk)  at (-6,-1.5) {$p_k$};
\node[leaf] (pj)  at (6,1.5) {$p_j$};
\node[leaf] (pl)  at (6,-1.5) {$p_\ell$};

% main internal edge, labelled pi (adjust "pos=" and "above/below" as needed)
\draw (u) -- (v)
  node[edgelabel, pos=0.55, above] {$\pi$};

% left side edges (longer now)
\draw (u) -- (piL);
\draw (u) -- (pk);

% subdivision points 1cm from v on its three incident edges
% (adjust "1cm" here if you want ticks closer/farther from v)
\coordinate (t_mid) at ($(v)!1.5cm!(u)$);
\coordinate (t_top) at ($(v)!1.5cm!(pj)$);
\coordinate (t_bot) at ($(v)!1.5cm!(pl)$);

% short edges from v to tiks, with labels
% tweak "pos=" and "above/below left/right" manually to move labels
\draw (v) -- (t_mid)
  node[edgelabel, pos=0.6, above]       {$e_{\text{mid}}$};

\draw (v) -- (t_top)
  node[edgelabel, pos=0.5, above ]  {$e_{\text{top}}$};

\draw (v) -- (t_bot)
  node[edgelabel, pos=0.5, below]  {$e_{\text{bot}}$};

% remaining parts of the edges
\draw (t_mid) -- (u);
\draw (t_top) -- (pj);
\draw (t_bot) -- (pl);

% small tiks at subdivision points
\foreach \x in {t_mid,t_top,t_bot}{
  \fill (\x) circle (1pt);
}

% dotted vertical line through v
\draw[densely dotted] (2,-2.2) -- (2,2.2);

% horizontal line on the right of v (lengthened to match new leaves)
\draw[densely dotted] (2,0) -- (7,0);

% region labels (adjust coordinates as desired)
\node[internal] at (4.5,1.8) {$B_{\text{top}}$};
\node[internal] at (4.5,-1.8) {$B_{\text{bot}}$};

\end{tikzpicture}

The remainder of the proof accumulates a series of facts under our contradiction assumptions, ultimately showing that every possible placement of the $u$-vertices, the $\gamma$-vertices, and $\omega$ leads to a $3$-induced matching across one of the edges $e_{mid}$, $e_{bot}$, or $e_{top}$.

\begin{fact}\label{fact:gammaB}
    All the $\gamma$-vertices $\gamma_i^q, \gamma_j^q, \gamma_k^q, \gamma_\ell^q$ of the quartet~$q$ lie in $B$.
    
\end{fact}

\begin{proof}
     Claim~\ref{claim:gammas} still holds, because $T$ can be turned into a branch decomposition of $G$ by deleting the $\omega$-leaf, which does not increase the mim-width of the decomposition.
\end{proof}

\begin{fact}\label{fact:uB}
     All $u$-vertices $u_i^q, u_j^q, u_k^q, u_\ell^q$ of the quartet~$q$ lie in $B$.
\end{fact}

\begin{proof}
    Suppose not, and assume without loss of generality that $u_i^q \notin B$ or $u_j^q \notin B$ (the cases $u_k^q \notin B$ or $u_\ell^q \notin B$ are symmetric). Then one of the edges $u_i^q u_j^q$ or $u_j^q p_j$ crosses the cut $(A_{mid}, B_{mid})$ (with $u_i^q \in A_{mid}$ in the first case). One of the edges of the path $p_k - u_k^q - u_\ell^q - p_\ell$ must also cross the cut. The vertex $\gamma_i^q$ is in $B = B_{mid}$ by Fact~\ref{fact:gammaB}. These two edges, together with the edge $p_i \gamma_i^q$, form an induced matching of size~3.
\end{proof}

\begin{fact}\label{fact:wb}
    The vertex $\omega$ is in $B$.
\end{fact}

\begin{proof}
    Assume for contradiction that $\omega \notin B$. Using Fact~\ref{fact:gammaB} and Fact~\ref{fact:uB}, all the $u$ and $\gamma$ vertices are in $B = B_{mid}$. In the cut $(A_{mid},B_{mid})$, the edges $\omega u_j^q$, $p_i\gamma_i^q$, and $p_k\gamma_k^q$ form an induced matching of size~3.
\end{proof}

\begin{fact}\label{fact:usameside}
    Each side $B_{top}$ and $B_{bot}$ contains at least one $u$-vertex.
\end{fact}

\begin{proof}
    Note that by Fact~\ref{fact:uB}, each $u$-vertex belongs to either $B_{top}$ or $B_{bot}$.  If all the four $u$-vertices lie on the same side, say $B_{top}$, then across the cut $(A_{top},B_{top})$ the edges $p_i u_i^q$, $p_k u_k^q$, and $p_\ell u_\ell^q$ form an induced matching of size~3. Likewise, not all four $u$-vertices can lie in $B_{bot}$.
\end{proof}

\begin{fact}\label{fact:gkorglbot}
    At least one of $\gamma_k^q,\gamma_\ell^q$ lies in $B_{bot}$. Symmetrically, at least one of $\gamma_i^q,\gamma_j^q$ lies in $B_{top}$.
    
\end{fact}
\begin{proof}
    Using Fact~\ref{fact:gammaB}, each $\gamma$-vertex is either in $B_{bot}$ or $B_{top}$. Assume $\gamma_k^q$ and $\gamma_\ell^q$ are both in $B_{top}$. Then one of the edges of the path $p_i - u_i^q - u_j^q - p_j$ must cross the cut $(A_{top},B_{top})$. Together with the edges $p_k\gamma_k^q$ and $p_\ell\gamma_\ell^q$, this yields an induced matching of size~3, a contradiction.
\end{proof}

Using Fact~\ref{fact:wb}, the vertex $\omega$ is either in $B_{bot}$ or in $B_{top}$. In the rest of the proof, we will assume without loss of generality that $\omega \in B_{bot}$ (to see why this is without loss of generality, notice that if this is not the case, then by symmetry we may swap the subscripts of $p_i$ and $p_k$, and the subscripts of $p_j$ and $p_\ell$ without affecting the previous arguments).

\begin{fact}\label{fact:uiujtop}
    The vertices $u_i^q$ and $u_j^q$ are both in $B_{top}$.
\end{fact}

\begin{proof}
    Consider $u_i^q$ first.  Using Fact~\ref{fact:uB}, $u_i^q$ is either in $B_{bot}$ or in $B_{top}$. Assume by contradiction that it is in $B_{bot}$. Across the cut $(A_{bot},B_{bot})$, the following 3 edges form an induced matching:
    \begin{itemize}[nosep, leftmargin=*]
        \item The edge $p_i u_i^q$.
        \item One of the edge $u_x^q \omega$ with $x \in \{j,k,\l\}$. At least one $u$ vertices is in $B_{top}$ by Fact~\ref{fact:usameside}.
        \item An edge joining a $\gamma$-vertex in $B_{top}$ to either $\gamma_k^q$ or $\gamma_\ell^q$ in $B_{bot}$. Such an edge exists using Fact~\ref{fact:gkorglbot}.
    \end{itemize}
    We may thus assume that $u_i^q \in B_{top}$.  
    
    Next assume for contradiction that $u_j^q \in B_{bot}$.  Then across the cut $(A_{bot}, B_{bot})$ we may form the induced matching containing the edges $p_j u_j^q$, $u_i^q \omega$, and as before an edge joining a $\gamma$-vertex in $B_{top}$ to either $\gamma_k^q$ or $\gamma_\ell^q$ in $B_{bot}$. 
\end{proof}

% \begin{fact}\label{fact:uitop}
%     The vertex $u_i^q$ is in $B_{top}$.
% \end{fact}

% \begin{proof}
%     Using Fact~\ref{fact:uB}, $u_i^q$ is either in $B_{bot}$ or in $B_{top}$. Assume by contradiction that it is in $B_{bot}$. Across the cut $(A_{bot},B_{bot})$, the following 3 edges form an induced matching:
%     \begin{itemize}[nosep, leftmargin=*]
%         \item The edge $p_i u_i^q$.
%         \item One of the edge $u_x^q \omega$ with $x \in \{j,k,\l\}$. At least one $u$ vertices is in $B_{top}$ by Fact~\ref{fact:usameside}.
%         \item An edge joining a $\gamma$-vertex in $B_{top}$ to either $\gamma_k^q$ or $\gamma_\ell^q$ in $B_{bot}$. Such an edge exists using Fact~\ref{fact:gkorglbot}.
%     \end{itemize}
% \end{proof}

% \begin{fact}\label{fact:ujtop}
%     The vertex $u_j^q$ is in $B_{top}$.
% \end{fact}

% \begin{proof}
%     Using Fact~\ref{fact:uB}, $u_j^q$ is either in $B_{bot}$ or in $B_{top}$. Assume by contradiction that it is in $B_{bot}$. Across the cut $(A_{bot},B_{bot})$, the following 3 edges form an induced matching:
%     \begin{itemize}[nosep, leftmargin=*]
%         \item The edge $p_j u_j^q$.
%         %\item One of the edge $u_x^q \omega$ with $x \in \{i,k,\l\}$. At least one $u$ vertices is in $B_{top}$ by Fact~\ref{fact:usameside}.
%         \item 
%         The edge $u_i^q \omega$ using Fact~\ref{fact:uitop}.
%         \item An edge joining a $\gamma$-vertex in $B_{top}$ to either $\gamma_k^q$ or $\gamma_\ell^q$ in $B_{bot}$. Such an edge exists using Fact~\ref{fact:gkorglbot}.
%     \end{itemize}
% \end{proof}

\begin{fact}\label{fact:gjorgktop}
    At least one of $\gamma_j^q,\gamma_k^q$ lies in $B_{top}$.
\end{fact}

\begin{proof}
    Using Fact~\ref{fact:gammaB}, each $\gamma$-vertex is either in $B_{bot}$ or $B_{top}$. Assume for contradiction that $\gamma_j^q$ and $\gamma_k^q$ are both in $B_{bot}$. Across the cut $(A_{bot},B_{bot})$, the following 3 edges form an induced matching: $p_j\gamma_j^q$, $p_k\gamma_k^q$, and $u_i^q \omega$ (we know that $u_i^q \in B_{top}$ using Fact~\ref{fact:uiujtop}).
    % \begin{itemize}
    %     \item $p_j\gamma_j^q$.
    %     \item $p_k\gamma_k^q$.
    %     \item $u_i^q \omega$. We know that $u_i^q \in B_{top}$ using Fact~\ref{fact:uitop}.
    % \end{itemize}
\end{proof}

\begin{fact}\label{fact:ukbot}
    The vertex $u_k^q$ is in $B_{bot}$.
\end{fact}

\begin{proof}
    Using Fact~\ref{fact:uB}, $u_k^q$ is either in $B_{bot}$ or in $B_{top}$. Assume by contradiction that it is in $B_{top}$. Across the cut $(A_{top},B_{top})$, the following 3 edges form an induced matching:
    \begin{itemize}[nosep, leftmargin=*]
        \item $p_i u_i^q$. We have $u_i^q  \in B_{top}$ by Fact~\ref{fact:uiujtop}.
        \item $u_\ell^q u_k^q$. We are assuming that $u_k^q \in B_{top}$. Moreover, we know that $u_i^q, u_j^q \in B_{top}$ (Fact~\ref{fact:uiujtop}), and that the $u$-vertices are not all in $B_{top}$ (Fact~\ref{fact:usameside}); therefore $u_\ell^q$ is in $B_{bot}$.
        \item An edge joining a $\gamma$-vertex in $B_{bot}$ to either $\gamma_j^q$ or $\gamma_k^q$ in $B_{top}$. Such an edge exists using Fact~\ref{fact:gjorgktop}.
    \end{itemize}
\end{proof}

\begin{fact}\label{fact:gkonlybot}
The vertex $\gamma_k^q$ is the only gamma vertex in $B_{bot}$.
\end{fact}
\begin{proof}
    Assume by contradiction that there is a $\gamma$-vertex other than $\gamma_k^q$ ib $B_{bot}$. Across the cut $(A_{bot},B_{bot})$, the following 3 edges form an induce matching:
    \begin{itemize}[nosep, leftmargin=*]
        \item $p_k u_k^q$. We know that $u_k^q$ is in $B_{bot}$ by Fact~\ref{fact:ukbot}.
        \item $u_i^q \omega$.  We know that $u_i^q$ is in $B_{top}$ by Fact~\ref{fact:uiujtop}.
        \item  An edge joining either $\gamma_i^q$ or $\gamma_j^q$ in $B_{top}$ (Fact~\ref{fact:gkorglbot}) to a $\gamma$-vertex in $B_{bot}$ that is not $\gamma_k^q$.
    \end{itemize}
\end{proof}

We can finally conclude the proof.  By combining Fact~\ref{fact:gammaB} and Fact~\ref{fact:gkonlybot}, we know that $\gamma_i^q$ and $\gamma_\ell^q$ are both in $B_{top}$. Across the cut $(A_{top},B_{top})$, the following 3 edges form an induced matching:
$p_i \gamma_i^q$, $p_\ell \gamma_\ell^q$, and $\omega u_j^q$ (since $u_j^q$ is in $B_{top}$ by Fact~\ref{fact:uiujtop}).
% \begin{itemize}
%     \item $p_i \gamma_i^q$.
%     \item $p_\ell \gamma_\ell^q$.
%     \item $\omega u_j^q$. $u_j^q$ is in $B_{top}$ by Fact~\ref{fact:ujtop}.
% \end{itemize}
%
This completes the proof: in all cases we obtain an induced matching of size~3.

\end{proof}

\begin{lemma}
    If there is a caterpillar satisfying the \textsc{UQC} instance $(Q,P)$, then $\lmw(G) = 2$.
\end{lemma}

\begin{proof}
    As we already observed, $\mimw(G) \geq 2$ because it contains induced cycles of length $6$, implying $\lmw(G) \geq 2$.  We focus on the upper bound.
    
    Let $T$ be a caterpillar satisfying $(Q, P)$, and let $\leq$ be a total order on $P$ that $T$ realizes. For each $p_i \in P$, define
\[
C_i \;=\; \{p_i\} \;\cup\; \{\,u_i^q \mid q \in Q,\; p_i \in q\,\} \;\cup\; \{\,\gamma_i^q \mid q \in Q,\; p_i \in q\,\},
\]
that is, $C_i$ consists of all vertices in $P \cup U \cup \Gamma$ whose index is $i$. We extend $\leq$ to a total order $\leq'$ on $P \cup U \cup \Gamma$ as follows: for all $p_i,p_j \in P$ with $p_i < p_j$, and for all $x \in C_i$, $y \in C_j$, we set $x <' y$; within each $C_i$, the order is arbitrary. Let $T'$ be the ternary caterpillar realizing $\leq'$. We prove that $\mimw_G(T') \leq 2$, i.e., $\mim_G(A_e,B_e) \leq 2$ for every edge $e$ of $T'$.

Because $T'$ is a caterpillar, it suffices to consider the cuts induced by the edges of its spine. Let $e$ be an edge of the spine, and let $(A_e,B_e)$ be the cut induced by $e$. Since $T'$ realizes $\leq'$, one side of the cut is an initial segment of $\leq'$. Thus, either for all $a \in A_e$ and $b \in B_e$ we have $a <' b$, or for all $a \in A_e$ and $b \in B_e$ we have $b <' a$.

Let $a_1b_1,a_2b_2,a_3b_3 \in E(G)$ be three edges with pairwise disjoint endpoints such that $a_1,a_2,a_3 \in A_e$ and $b_1,b_2,b_3 \in B_e$. Suppose, for contradiction, that these three edges form an induced matching in $G[A_e,B_e]$.

First, note that at most one of the six vertices lies in $P$. Indeed, a vertex $p_i \in P$ is adjacent only to vertices of $C_i$. If $p_i$ were among the six vertices, 
then $e$ would separate $p_i$ from another element of $C_i$; 
%then $C_i$ would meet both $A_e$ and $B_e$; 
by the construction of $\leq'$, this can happen for at most one index $i$.

We say that a vertex $v$ of $G$ is of \emph{type $q$} for a quartet $q$ if $v$ is one of the eight vertices in $U \cup \Gamma$ corresponding to $q$ (i.e., the type is the superscript $q$). Suppose $a_1b_1$ is an edge with $a_1$ of type $q$ and $b_1$ of type $q' \neq q$. In our construction, the only vertices not adjacent to $a_1$ are those in $P$ or of type $q$; hence $b_2$ and $b_3$ must each lie in $P$ or be of type $q$. Symmetrically, the only vertices not adjacent to $b_1$ are those in $P$ or of type $q'$, so $a_2$ and $a_3$ must each lie in $P$ or be of type $q'$. If, without loss of generality, $a_3$ is the only possible vertex possibly in $P$, then $a_2$ is of type $q'$ and $b_3$ is of type $q$, which forces $a_2b_3 \in E(G)$, contradicting that the three edges form an induced matching.

Therefore, at most one of the six vertices in the induced matching is in $P$ and the remaining vertices in $U \cup \Gamma$ all have the same type $q=[p_ip_j \mid p_k p_\ell]$. The three matching edges must then be of the following kinds: $\gamma$–$\gamma$, $u$–$u$, $p$–$u$, or $p$–$\gamma$.

We cannot have two $\gamma$–$\gamma$ edges, since $\Gamma$ is a clique and there would be an extra edge between the $\gamma$-endpoints, violating the induced property.

Assume there are two $u$–$u$ edges. The only possibilities are $u_i^q u_j^q$ and $u_k^q u_\ell^q$. Without loss of generality, let $u_i^q,u_k^q \in A_e$ and $u_j^q,u_\ell^q \in B_e$. By the definition of $\leq'$, this implies $u_i^q,u_k^q <' u_j^q,u_\ell^q$, hence $p_i,p_k < p_j,p_\ell$ in $\leq$, which contradicts that $T$ satisfies the quartet $q$.

Since at most one of the six vertices lies in $P$, it follows that \emph{exactly one} matching edge is of type $u$–$u$ (say $a_1b_1$), \emph{exactly one} is of type $\gamma$–$\gamma$ (say $a_2b_2$), and the remaining one (say $a_3b_3$) is either $p$–$u$ or $p$–$\gamma$.

Without loss of generality, let $a_3=p_i \in P$. If $b_3=\gamma_i^q$, then $b_3$ is adjacent to $a_2$ (both are $\gamma$-vertices), so the matching is not induced. Otherwise, if $b_3=u_i^q$, then $a_1b_1$ must be $u_k^q u_\ell^q$. Since $e$ separates $p_i$ from $u_i^q$, the cut lies inside $C_i$; together with $u_k^q u_\ell^q$ crossing the cut, this forces $p_i$ to lie between $p_k$ and $p_\ell$ in the order $\leq$, contradicting that $T$ satisfies the quartet $q$. This contradiction shows that no induced matching of size $3$ can appear in $G[A_e,B_e]$, and hence $\mim_G(A_e,B_e) \leq 2$ for every spine edge $e$ of $T'$. Therefore, $\mimw_G(T') \leq 2$.
\end{proof}

\begin{lemma}
    If there is a caterpillar satisfying the \textsc{UQC} instance $(Q,P)$, then $\mimw(H) = 2$.
\end{lemma}

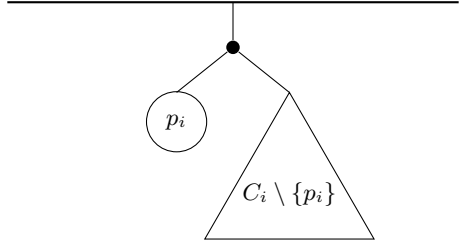
\begin{figure}[ht]
\centering
\begin{tikzpicture}[
    scale=1.5,
    every node/.style={font=\small},
    spine/.style={thick},
    root/.style={circle, fill, inner sep=1.8pt},
    leaf/.style={circle, draw, minimum size=0.8cm}
]

% horizontal spine
\draw[spine] (-2,0) -- (2,0);

% root slightly below the spine
\node[root] (r) at (0,-0.4) {};
\draw (0,0) -- (r);

% leaf p_i: NORTH at (-0.5,-0.8)
% center is 0.4 cm below (-0.5,-0.8)
\node[leaf] (pi) at (-0.5,-1.06) {$p_i$};

% endpoint for the edge to C_t (symmetric with p_i north)
\coordinate (Ct) at (0.5,-0.8);

% symmetric edges
\draw (r) -- (-0.5,-0.8);   % touching the north of the leaf circle
\draw (r) -- (Ct);

% larger subtree triangle
\coordinate (Cl) at (-0.25,-2.1);
\coordinate (Cr) at (1.25,-2.1);
\draw (Ct) -- (Cl) -- (Cr) -- cycle;

% label inside the triangle
\node at (0.5,-1.7) {$C_i \setminus \{p_i\}$};

\end{tikzpicture}
\caption{Expansion of $p_i$ in $T$ into the subtree corresponding to $C_i$ in 
$T'$.}
\label{fig:rooted-tree-spine}
\end{figure}

\begin{proof}
   Again, $\mimw(H) \geq 2$ follows from the existence of induced cycles of length 6 and we show that $\mimw(H) \leq 2$.
   
   Let $T$ be a caterpillar satisfying $(Q, P)$. For each $p_i \in P$, define
\[
C_i \;=\; \{p_i\} \;\cup\; \{\,u_i^q \mid q \in Q,\; p_i \in q\,\} \;\cup\; \{\,\gamma_i^q \mid q \in Q,\; p_i \in q\,\},
\]
that is, $C_i$ consists of all vertices in $P \cup U \cup \Gamma$ whose index is $i$.

We define a new tree $T'$ from $T$ by adding the leaf $\omega$ anywhere on the caterpillar. For each $i$, replace the leaf $p_i$ by a rooted subtree whose leaves are exactly the vertices of $C_i$. The leaf $p_i$ is attached directly to the root of this subtree, and all the other leaves lie on the other branch (see Figure~\ref{fig:rooted-tree-spine}).

We prove that $\mimw_G(T') \le 2$, i.e., $\mim_G(A_e,B_e) \le 2$ for every edge $e$ of $T'$. Fix any edge $e$ of $T'$.

Let $a_1b_1,a_2b_2,a_3b_3 \in E(G)$ be three edges with pairwise disjoint endpoints such that $a_1,a_2,a_3 \in A_e$ and $b_1,b_2,b_3 \in B_e$. Suppose, for contradiction, that these three edges form an induced matching in $G[A_e,B_e]$.

First note that none of the six endpoints lies in $P$. Indeed, each $p_i\in P$ is adjacent exactly to the vertices of $C_i$, and by the placement of $p_i$ in $T'$, any cut $(A_e,B_e)$ with $p_i\in A_e$ that separates $p_i$ from some vertex of $C_i$ must be of one of two forms. Either $A_e=\{p_i\}$, then there is no matching of size $3$; or $B_e \subseteq C_i$ because $e$ is in the subtree containing $C_i \setminus \{p_i\}$, then, since $p_i$ is adjacent to every vertex of $C_i$, the matching would not be induced.

As in the previous proof, we say that a vertex $v$ of $G$ is of \emph{type $q$} for a quartet $q$ if $v$ is one of the eight vertices in $U \cup \Gamma$ corresponding to $q$. Suppose $a_1b_1$ is an edge with $a_1$ of type $q$ and $b_1$ of type $q' \neq q$. In our construction of $G$, the only vertices nonadjacent to $a_1$ are those in $P$, those of type $q$, and possibly $\omega$. Since no endpoint is in $P$, it follows that $b_2$ and $b_3$ must be of type $q$ or be $\omega$; symmetrically, $a_2$ and $a_3$ must be of type $q'$ or be $\omega$.  If, without loss of generality, $a_3$ is the only vertex that might be $\omega$, then $a_2$ is of type $q'$ and $b_3$ is of type $q$, which forces $a_2b_3 \in E(G)$, contradicting that the three edges form an induced matching.

Hence at most one of the six endpoints is $\omega$, and all remaining endpoints in $U \cup \Gamma$ belong to a common type (say $q$). Thus the three matching edges must be among the following kinds: $\gamma$–$\gamma$, $u$–$u$, or $\omega$–$u$.

We cannot have two $\gamma$–$\gamma$ edges, since $\Gamma$ is a clique, which would create an extra edge between $\gamma$-endpoints and violate the induced property.

Assume there are two $u$–$u$ edges. The only possibilities are $u_i^q u_j^q$ and $u_k^q u_\ell^q$. Without loss of generality, let $u_i^q,u_k^q \in A_e$ and $u_j^q,u_\ell^q \in B_e$. This implies that in $T'$, the $u_i^q-u_j^q$ path intersects the $u_k^q-u_\ell^q$ path (in particular, they intersect at the ends of $e$).  By our construction of $T'$ from $T$, this in turn implies that in $T'$, the $p_i-p_j$ path intersects the $p_k-p_\ell$ path, and thus $T'$ does not satisfy $q$.  Since our transformation from $T'$ to $T$ does not alter its quartets but only adds new ones, this contradicts that $T$ satisfies $q$.

Moreover, we cannot have simultaneously a $u$–$u$ edge and an $\omega$–$u$ edge, because $\omega$ is adjacent to all $u$-vertices, creating a forbidden cross-edge and destroying the induced property. Since there is at most one $\omega$–$u$ edge, any third edge would have to be either $u$–$u$ or $\gamma$–$\gamma$, both of which have been ruled out. Therefore no induced matching of size~3 can exist, a contradiction.

This proves that $\mim_G(A_e,B_e) \le 2$ for every edge $e$ of $T'$, and hence $\mimw_G(T') \le 2$.
\end{proof}

\thmmim*

\begin{proof}
    We proceed as in Theorem~\ref{thm:sim}.  Take an instance $I$ of \textsc{Betweenness} and let $(Q, P)$ be the \textsc{UQC} instance obtained from Theorem~\ref{thm:steel}.  Then let $G$ and $H$ be obtained from $(Q, P)$ as described above.  If $I$ is a YES-instance, then $(Q, P) \in \textsc{UQC}$ and by Theorem~\ref{thm:steel} some ternary caterpillar satisfies $Q$.  Then by Proposition~\ref{prop:mimw-lmimw} we have $\lmw(G) = 2$ and $\mimw(H) = 2$.  
    If $I$ is a NO-instance, then $(Q, P) \notin \textsc{UQC}$ and by Proposition~\ref{prop:mimw-lmimw} we have $\lmw(G) \geq 3$ and $\mimw(H) \geq 3$.  Thus $I$ is a YES-instance of \textsc{Betweenness} if and only if $\lmw(G) \leq 2$, and if and only if $\mimw(H) \leq 2$, which proves the NP-hardness of our two problems (and membership in NP is easy).  

    As for the ETH lower bound, we observe that $|V(G)| = |P| + 8|Q|$ and $|V(H)| = |P| + 8|Q| + 1$.  Since $|V(G)|$ and $|V(H)|$ are linear in $|P| + |Q|$, Theorem~\ref{thm:uqceth} implies that no time $2^{o(n)}$ algorithm can recognize $\mimw$ 2 and $\lmw$ 2 graphs under the ETH.
\end{proof}

\section{Open problems}
Many questions remain open. The most important one is whether graphs of mim-width~1 and graphs of linear mim-width~1 can be recognized in polynomial time.

Our result also does not imply that deciding whether (linear) $\mimw(G) \le c$ is NP-hard for every fixed constant $c > 2$. One possible direction is to address the following question: Is there a graph operation that transforms any graph $G$ into a graph $G'$ satisfying $\mimw(G') = \mimw(G) + 1$? Alternatively, is there a graph operation that multiplies the mim-width of any graph by a constant factor? The same question remains open for other width parameters as well.

Another question is whether \textsc{UQC} can be used to prove hardness of other graph-to-tree representations?

Finally, can one obtain inapproximability results, for instance via a reduction that yields an arbitrarily large approximation gap?

\printbibliography

\end{document}